\documentclass{easychair}

\usepackage{doc}

\usepackage{dafny}
\usepackage{color}
\usepackage{graphicx}
\usepackage{algpseudocode}

\newcommand{\bs}{}
\newcommand{\as}{}

\newcommand{\scode}[1]{\lstinline[basicstyle=\normalsize\sffamily]!#1!}
\newcommand{\code}[1]{\lstinline[basicstyle=\normalsize\sffamily]!#1!}

\newcommand{\mybox}[1]{{\begin{center}\fbox{\begin{minipage}{0.97\linewidth}#1\end{minipage}}\end{center}}}

\newcommand{\algname}[1]{\textsc{#1}}

\newtheorem{mydef}{Definition}

\newtheorem{myprop}{Proposition}

\title{DAReing to reduce the annotation overheads of verified programs%
\thanks{This work has been supported by EPSRC grants EP/M018407/1 and  EP/N014758/1.}}

\author{
Gudmund Grov \and Duncan Cameron and \and L\'eon McGregor\inst{1}
}

\institute{
Heriot-Watt University, UK \\ \email{\{G.Grov,dac31,lm356\}@hw.ac.uk}}

\authorrunning{Grov, Cameron and McGregor}

\titlerunning{DAReing to reduce the annotation overheads of verified programs}

\begin{document}

\maketitle



\begin{abstract}
Modern program verifiers use the same uniform program text to both specify and implement programs. The program text is also used to provide the necessary guidance to ensure that the program satisfies its specification. The amount of guidance required is often called \emph{the annotation overhead}. This can be high and is often seen as a hindrance for  wider use of program verifiers, as 
development time is increased and the guidance may obfuscate the program text. In this paper we introduce the DARe tool, which automatically removes as much unnecessary guidance as possible for the Dafny program verifier \cite{Leino10}. The tool is integrated with the Dafny IDE \cite{leino2014dafnyIDE}. 
To evaluate DARe, we apply it to $252$ programs from the Dafny library \cite{DafnyLib} and analyse the degree to which it is able to remove unnecessary guidance. Our results are very encouraging as a staggering $88\%$ of the guidance can be removed.
\end{abstract}

\bs{}
\section{Introduction}
\as{}



Users of modern \emph{auto-active} program verifiers, such as Spec\# \cite{Barnett04}, VCC \cite{Cohen09}, Verifast \cite{Jacobs10}, Dafny \cite{Leino10}, and  SPARK 2014 \cite{mccormick2015}, use the program text to encode: (i) the program specification, (ii) the program implementation, and (iii) the proof guidance to ensure that the program satisfies the specification. 

The advantage of these techniques, compared with interactive theorem provers such as Isabelle \cite{Paulson90} or Coq \cite{bertot04book}, is that both programs and proofs are developed in the same language, and using constructs more
familiar to software engineers. Thus, a developer does not need to learn to use an additional system.

The ratio of the number of lines of specification and annotations (LoA) to the number of lines of code (LoC) is called the \emph{annotation overhead}.
For non-trivial programs this ratio can be high, for example \cite{Penninckx12} reported on $4.8$ LoA per LoC. Such a high degree of annotations may obfuscate 
the program text and act as a hindrance for mainstream uptake of the technique beyond niche markets. There are at least two reasons why stated annotation overheads may be too high:
\begin{enumerate} 
\item Annotations are normally added incrementally to the program text until a proof is found. Previous increments may or may not have helped to progress the proof. Instead of attempting the (possibly time consuming and tedious) task of manually ``tidying up'' these annotations, a developer may be satisfied that the job is done. This scenario will be illustrated in \S \ref{sec:dafny}.
\item The underlying verifier is improved such that guidance that used to be required is no longer needed. Periodic ``tidying'' annotations of existing libraries may be too time consuming for a developer, in particular if this has to be done manually.
\end{enumerate}
In this paper we present the \emph{DARe} tool (\underline{D}ead \underline{A}nnotations \underline{Re}moval), which automatically removes annotations that are superfluous (or dead\footnote{Inspired by the term used for ``dead code'' we call them ``dead annotations''.}) for the Dafny program verifier. To investigate how useful this tool is, and thus the degree to which annotations are superfluous, we empirically evaluate the program using examples from the Dafny library and test suite \cite{DafnyLib}. For usability purposes, we investigate the impact DARe has on runtime, and 
embed it in the Dafny IDE for Visual Studio \cite{leino2014dafnyIDE}. This requires instant feedback and we therefore investigate and compare algorithms that are either simple, complete or fast.


\lstset{numbers=left,numbersep=5pt,stepnumber=1}

\bs{}
\section{Proofs \& programs in Dafny}\label{sec:dafny}
\as{}


Dafny is an imperative, object-oriented and functional programming language that has been designed for verification. Properties that have to be satisfied are commonly specified by \emph{contracts}: given a precondition that a method can assume, it must guarantee that a given postcondition holds. Special proof constructs, together with standard programming constructs, are used to provide guidance to ensure that the contract holds.  To verify a program, Dafny translates it into an \emph{intermediate verification language} (IVL)\footnote{An IVL can be seen as a layer to ease the process of generating new program verifiers.} called \emph{Boogie} \cite{Barnett06}. From Boogie a set of \emph{verification conditions} (VCs) is generated and applied to the \emph{Z3 SMT solver} \cite{Moura08}. If it fails, then the failure is translated back to the Dafny code, via Boogie.

\subsection{Assertions \& lemma calls} The simplest form of annotation is the use of \code{assert} statements in the code. An assertion must be verified, and will subsequently be used by the prover. Dafny also supports a \emph{ghost state}.
From a programming point of view, working in the ghost state is exactly the same the non-ghost state. The difference is that it is only used for verification and will not be compiled. To illustrate,
a ghost variable \scode{v} is declared as \scode{ghost var x}, and then used like a normally variable.

A \code{lemma} is a ghost method, i.e. a method that is in the ghost state (including all its local variables). It is therefore only used for verification purposes.
 Lemmas are used for more complex assertions that require guidance in order to be verified.  A lemma contract specifies the desired property, while the method body encodes the proof guidance. To illustrate, consider the following lemma and associated proof by contradiction \cite{Grov2016FM}:
\begin{lstlisting}
lemma set_inter_empty_contr(A:set<int>, B:set<int>, x:int)
 requires x in A && A * B == {}
 ensures !(x in B)
{
  if x in B {
    assert x in A * B;
    set_eq_simple(A*B,{},x);
    assert x in {};
    assert false; 
  }
}
\end{lstlisting} 
The lemma states that if a given variable \scode{x} is a member of the set \scode{A}, and the sets \scode{A} and \scode{B} are disjoint (their intersection \scode{A*B} is empty),
then \scode{x} cannot be in set \scode{B}. 

\lstset{numbers=none}

We will now illustrate how this proof was found, and use \scode{(err)...(enderr)} to highlight where Dafny reports that the error is. The proof in the lemma body started by setting up the contraction using an  \scode{if}-statement: 
\begin{lstlisting}
 if x in B { 
   (err) (enderr) 
 }
\end{lstlisting}
Dafny will translate the lemma with this body to Boogie, which will generate a set of VCs that are applied to the Z3 SMT solver. One of the VCs fails, and Boogie relates that to the 
post-condition not being satisfied for the case where \scode{x in B}. The error is then highlighted inside this case within the Dafny IDE, as illustrated above. As only this case is
highlighted we know that Dafny can prove the property when \scode{!(x in B)}, which is unsurpricing considering that this is the same as the postcondition.
For the problematic \scode{x in B} case, more guidance is required in the program text. 

We then continued developing the proof by asserting that \scode{x} is in the intersection of \scode{A} and \scode{B} (\scode{x in A * B}), which is trivial to prove as we know 
 \scode{x in A} fromt a precondition and that \scode{x in B} from the condition of the \scode{if} statement. We also add the contradiction we are deriving
(line 8 of \scode{set\_inter\_empty\_contr}). Dafny then highlights the following mistake:
 \begin{lstlisting}
 if x in B {
    assert x in A * B;
  (err) assert false; (enderr) 
 }
\end{lstlisting} 
Here, Dafny report that the last assertion (\scode{assert false}) cannot be proven. It does not complain about the first assertion, which means that it can be verified. 
Moreover, it no longer complain about the postcondition, meaning that it can be derived from the last assertion (which we still need a proof of).
We then add the assertion from line 5, from which we can derive \scode{false}:
 \begin{lstlisting}
 if x in B {
    assert x in A * B;
  (err) assert x in {}; (enderr) 
    assert false; 
 }
\end{lstlisting} 
We now see that Dafny can infer the contradiction, but needs help to show that \code{x in \{\}}. To bridge this line to the previous line,
we add a call to lemma \scode{set_eq_simple} (line 6). This lemma states that if \scode{x} is in \scode{A*B} then \scode{x} is in \scode{\{\}}:
 \begin{lstlisting}
lemma set_eq_simple(A : set<int>, B:set<int>, x :int)
  requires x in A && A == B
  ensures x in B
{ }
\end{lstlisting} 
This lemma is proven automatically (as the body is empty). The proof is now complete.

\lstset{numbers=left}

\subsection{Loop invariants \& variants} Loop invariants express properties that have to hold throughout the execution of a loop, and are normally required for any program that contains a loop.
 Dafny methods and loops must terminate (with one exception discussed in \S \ref{sec:dare}). For simple cases, Dafny can prove this automatically; for more complex cases, the developer must provide an expression, called a \emph{variant}\footnote{Other names for variants include rank functions or decrease clauses.}, which for each step has to
decrease towards a lower bound. 

To illustrate, consider the following implementation of binary search over an array, which returns the index where the value is, or returns \scode{-1} when the search fails to find the value \cite{DafnyLib}:
\begin{lstlisting}
method BinarySearch(a:array<int>,value:int) returns (index:int)
 requires a != null && 0 <= a.Length && sorted(a)
 ensures 0 <= index ==> index < a.Length && a[index] == value
 ensures index < 0 ==> forall k :: 0 <= k < a.Length ==> a[k] != value
{
  var low, high := 0, a.Length;
  while low < high
    decreases high - low
    invariant 0 <= low && low <= high && high <= a.Length
    invariant forall i :: 0 <= i < a.Length && !(low <= i < high) ==> a[i] != val
  { 
  	var mid := (low + high) / 2;
    if a[mid] < value { 
    	low := mid + 1; 
    } else if value < a[mid] { 
    	high := mid; 
    } else { 
    	return mid; 
    }
  }
  return -1; 
}
\end{lstlisting} 
The contract specifies the desired properties, and a \scode{while} loop is used to implement it. Lines 8 and 9 have a loop invariant that the developer has to provide to help the prover verify that the implementation indeed satisfies the contract. In addition, a variant is provided (line 7), which states that \scode{high - low}, meaning that the difference between \scode{high} and \scode{low} is reduced for each step of the loop.

\subsection{Calculations} The proof of \scode{set_inter_empty_contr} is a linear sequence of arguments: from \scode{x in A * B} it implies that \scode{x in \{\}}, which implies \scode{false}. 
The first implication is justified by a call to \scode{set_eq_simple(A*B,\{\},x)}. 

Dafny supports such \emph{calculational proofs} via the \scode{calc} statement \cite{Leino13a}, which is inspired by Dijkstra's way of writing calculuations \cite{Dijkstra90}. Using calculations, the body of \scode{set_inter_empty_contr} can instead be written as:
\begin{lstlisting}
 if x in B {
  calc {
    x in A * B;
  ==> { set_eq_simple(A*B,{},x);}
    x in {};  
  ==>
    false; 
  }
 }
\end{lstlisting} 
On line 1 a calculation is started, while line 2 gives the starting statement that \scode{x in A * B}. This implies (line 3) that \scode{x in \{\}} (line 4), justified by the lemma call \scode{set_eq_simple(A*B,\{\},x)} (line 3).
Then, \scode{x in \{\}} (line 4) implies (line 5) \scode{false} (line 6). This last step can be proven automatically, so Dafny does not require any further justification/hints.

Note that there are multiple ways of writing such calculations. For example, we can
state that by default, arguments are chained by an implication, and thus omit it for each step:
\begin{lstlisting}
 if x in B {
  calc ==>{
   x in A * B;
   { set_eq_simple(A*B,{},x);}
   x in {};  
   false; 
  }
 }
\end{lstlisting} 
We return to calculations in \S \ref{sec:dare}.

\lstset{numbers=none}

\bs{}
\section{The DARe approach to dead annotation removal}\label{sec:dare}

\subsection{A motivating example}

Let us return to the proof of the \scode{set_inter_empty_contr} lemma, and recall how we worked with the system by incrementally adding annotations until a proof was found.
At the end of the proof, we do not know which steps were required, with the exception of the last step that was added. One way to find the steps that were required is to remove the added annotations
one-by-one and see if the proof still holds. If we do this we will see that none of the annotations except the last lemma call were needed (i.e. lines 4,6 and 7). The body of the lemma can therefore be simplified to:
\begin{lstlisting}
 if x in B { 
   set_eq_simple(A*B,{},x); 
 }
\end{lstlisting} 
In the rest of the paper we describe and evaluate the DARe tool to automate this process. First we will describe what constitutes an annotation.

\subsection{Specification vs program vs proof guidance}

When removing annotations manually, as we did above, it is normally straightforward for the developer to separate
annotations that are only used to guide the prover (which can be removed) from specification elements (which should be kept). However, it is 
not that easy to formalise this relationship so that it can be understood by a computer. For example,
Dafny uses many of the same constructs for both proof and programming; e.g. the proof of  the \scode{set_inter_empty_contr} lemma
is really just a program. 
Moreover, annotations are sometimes used to help explain a correctness argument, which is desirable to keep.
In order to develop a tool that automatically removes dead annotations, we have to decide upon what constitutes an annotation that can be removed.  We detail our choices next. Some may disagree with them, and in the future we plan to empirically evaluate these choices (\S \ref{sec:concl}). Note that the DARe tool (\S \ref{sec:tool}) is modular and can easily be updated to incorporate such changes.

We consider all \emph{assertions}, \emph{lemma calls}, \emph{invariants} and \emph{variants} to be \emph{annotations} that can be removed. From our experience, this is normally the case. We keep the lemmas, even if all calls to them are removed.
We do not consider common programming constructs, such as variable declarations, assignments, loops and conditionals, 
 to be proof guidance, even if they are in the ghost state. These are therefore kept. This is something we will revisit in the future. We consider \emph{contracts} to be specification. This is not always the case: for example, one can implement iteration using recursive functions. Here, a ``loop invariant'' is encoded using preconditions and postconditions.
 In this case, contracts become proof guidance, but we cannot (syntactically) separate them from the actual contracts that specify the program.  In the future we may extend our work to weaken preconditions\footnote{Dual to weakening a precondition is to strengthen a postcondition. However, this does not fit our ethos of simplifying the program text by removing elements.}, which will require global analysis. Finally, calculations are considered to be annotations used for proof guidance.

We can now define what constitutes \emph{annotations} (for proof guidance) and \emph{dead annotations} in our context:
\begin{mydef}[Annotation] \label{def:annot}
An \emph{annotation} is a lemma call, a variant, an invariant, an assertion, a calculation, a calculation step or a calculation hint.
\end{mydef}
\begin{mydef}[Dead annotation]\label{def:dead:annot}
A \emph{dead annotation} is an annotation, or a conjunct of an annotation, that can be removed from a verified Dafny method as long as the method
still verifies after the annotation has been removed.
\end{mydef}

\subsection{Overall approach}
In the next sections we will provide details of the algorithms, architecture and implementation of the DARe tool.
DARe essentially works by removing annotations from the program text and checks 
if the program still verifies. Here, we will describe how the various elements are handled and illustrate it with examples from \cite{DafnyLib}\footnote{See \S \ref{sec:eval:data} for details about the examples used.}.
This will give an indication of the type of simplifications that can be done by DARe and those that are not handled. Note that all simplification has been done by the tool described in \S \ref{sec:tool}.

\subsection{Lemma calls} The simplest cases are lemma calls. They are simply removed from the program text. To illustrate, consider the implementation of Fibonacci from \cite{DafnyLib}:
\begin{lstlisting}
function Fib(n: nat): nat
{
  if n < 2 then n
  else Fib(n - 2) + Fib(n - 1)
}
\end{lstlisting}
The property that \scode{n} is divisible by \scode{3} is equivalent to that its Fibonacci number is divisible by \scode{2}, is then expressed and proven as follows:
\begin{lstlisting}
lemma FibLemma(n: nat)
  ensures Fib(n) % 2 == 0 <==> n % 3 == 0
  decreases n
{
  if n < 2 {
  } else {
    FibLemma(n - 2);
    FibLemma(n - 1);
  }
}
\end{lstlisting}
The proof follows the same structure as the definition of \scode{Fib} with two recursive calls. However, this is something Dafny will do automatically, so DARe is able to simplify this to\footnote{
In this case, Dafny is actually able to prove the lemma fully automatically, so the body could be removed. We have taken a conservative approach to our definition of annotations (see Definition \ref{def:annot}),
meaning we do not attempt to remove such programming constructs. This is however something we may reconsider in the future (see \S \ref{sec:concl}).
}:
\begin{lstlisting}
lemma FibLemma(n: nat)
  ensures Fib(n) % 2 == 0 <==> n % 3 == 0
{
  if n < 2 {
  } else {
  }
}
\end{lstlisting}

\subsection{Assertions} 

As with lemma calls, DARe first try to remove assertions. This was illustrated for  \scode{set_inter_empty_contr} above, where three assertions were removed.

If it is not possible to remove an assertion, then we we do not want to separate the equivalent cases illustrated by  
\begin{lstlisting}
  assert A && B;
\end{lstlisting}
and
\begin{lstlisting}
  assert A;  
  assert B;
\end{lstlisting}
For the first case, conjunctions are handled by first breaking the expressions up, then removing them one at a time, and re-combining them at the end. 

\subsection{Variants} 

DARe attempts to remove variants. This is illustrated above for \scode{FibLemma}, with further examples below. 

\subsubsection{`Wild-card' variants (\scode{decreases *})}  

There is a special case for what we call \emph{wild-card variants} (\scode{decreases *}). These are used for programs that may not terminate.
For example, when working with infinite sequence (using co-inductive datatypes) or 
control systems with a top-level non-terminating loop. There may also cases where a developer may decide to delay, or even ignore, a proof of termination. 

To support such cases, a wild-card (\scode{*}) can be used in the decrease clause to tell Dafny to ignore proving termination. For example, the following non-terminating
methods is valid:
\begin{lstlisting}
method InFinite()
 decreases *
{
 while true
  decreases *
 { }
}
\end{lstlisting}
Note that the wild-card has to be used consistently -- if a method has a potentially non-terminating loop then the method needs to be declared to be potentially non-terminating, as illustrated above.

One could potentially use DARe to remove unnecessary wild-cards. However, in our data sets there were very few examples of this, and those that were there could not be further simplified. 
We will therefore not give this type of annotation much attention.

\subsection{Invariants} 

Invariants are handled the same way as assertions. To illustrate, consider again the \scode{BinarySearch} example. Firstly, the variant can be removed. The first two conjuncts of 
\begin{lstlisting}
  invariant 0 <= low && low <= high && high <= a.Length
\end{lstlisting}
can also be removed. This is achieved by DARe by treating it as:
\begin{lstlisting}
  invariant 0 <= low 
  invariant low <= high
  invariant high <= a.Length
\end{lstlisting}
and then remove the first two invariants. As a result we are left with:
\begin{lstlisting}
 while low < high
  invariant high <= a.Length
  invariant forall i :: 0 <= i < a.Length && !(low <= i < high) ==> a[i] != value
\end{lstlisting}
Another example is the following implementation of a method \scode{Max}, then returns the index of the largest element in a given array:
\begin{lstlisting}
method Max(a: array<int>) returns(t: int) 
  requires a != null
  requires a.Length > 0
  ensures forall i :: 0 <= i < a.Length ==> a[i] <= t
{
  var i: int := 1;
  var max: int := 0;
  while(i < a.Length)
   invariant i <= a.Length
   invariant max < a.Length
   invariant i > 0 && max >= 0
   invariant forall j :: 0 <= j < i ==> a[j] <= a[max]
   decreases a.Length-i
  {
    if(a[i] > a[max]) { max := i; }
    i := i + 1;
  }
  return a[max];
}
\end{lstlisting}
Here, the variant 
\begin{lstlisting}
  decreases a.Length-i
\end{lstlisting}
and the invariant
\begin{lstlisting}
  invariant i > 0 && max >= 0
\end{lstlisting}
can be removed, giving the following \scode{while}-loop header:
\begin{lstlisting}
 while(i < a.Length)
  invariant i <= a.Length
  invariant max < a.Length
  invariant forall j :: 0 <= j < i ==> a[j] <= a[max]	
\end{lstlisting}

\subsection{Calculations} 

We have seen an example of a calculation that was used to chain together implications. Such proofs consists of chain together a sequence of facts, possibly combined with an operator and/or a justification/hint.

The simplest case for DARe is to remove complete calculations, but this may not be possible. For those cases, we try to simplify them. Calculations are used to write out the chain of reasoning in order to connect
the first element of the calculation with the last element. For example, in the calculation of \S \ref{sec:dafny} it is shown how to derive \scode{false} (last element) from \scode{x in A * B} (first element). We
therefore keep the first and last element, and try to remove all steps and hints between. This is achieved by splitting them up into individual elements, remove as many as possible, and then combine the remaining elements.  For the calculation example in \S \ref{sec:dafny}, we can remove one element (\scode{x in \{\}}), creating the simplified calculation:
\begin{lstlisting}
 if x in B {
  calc ==>{
   x in A * B;
   { set_eq_simple(A*B,{},x);}
   false; 
  }
 }
\end{lstlisting} 

For a more involved example, consider the following representation of Peano arithmetic, with a recursive definition of addition\footnote{The example is taken by a course on `'Verified programming in Dafny' by Will Sonnex and Sophia Drossopoulou.}:
\begin{lstlisting}
datatype natr = Zero | Succ(natr)

function add(x: natr, y:natr): natr
{
  match x
   case Zero => y
   case Succ(x') => Succ(add(x', y))
}
\end{lstlisting} 
The following lemma proves that addition commutes: 
\begin{lstlisting}
lemma prop_add_comm(x: natr, y: natr)
	ensures add(x,y) == add(y,x)
{
  match x {
   case Zero =>
    calc {
      add(Zero, y);
     ==
      y;
     == { prop_add_Zero(y); }
      add(y, Zero);
    }
   case Succ(x') =>
    calc {
      add(x, y);
     == {assert x == Succ(x');}
      add(Succ(x'), y);
     ==
      Succ(add(x', y));
     == {prop_add_comm(x', y);}
      Succ(add(y, x'));
     == {prop_add_Succ(y, x');}
      add(y, Succ(x'));			
    }
  }
}
\end{lstlisting}
As the \scode{add} function is defined recursively over the first argument, the proof follows the same structure by a
case analysis (\scode{match} statement) of the first arguments. As we will see, this is the first step of a proof by \emph{structural induction} \cite{burstall69}.

For each case, the proof is achieved by a calculation. For the \scode{Zero} case, the proof has two steps: first that we can remove adding \scode{Zero} and the second
that it can be added to the second argument, which completes the proof. The second step is justified by a separate lemma, which is proven automatically:
\begin{lstlisting}
lemma prop_add_Zero(x: natr)
	ensures add(x, Zero) == x
{}
\end{lstlisting}
The other case (\scode{Succ}) follows a similar strategy by first moving \scode{Succ} outside \scode{add}, apply the induction hypothesis \scode{prop_add_comm(x', y)} to
commute the arguments and the move \scode{Succ} into the second argument. The last step justified by the following lemma, that does not require any proof guidance:
\begin{lstlisting}
lemma prop_add_Succ(x: natr, y: natr)
	ensures Succ(add(x,y)) == add(x, Succ(y))
{}
\end{lstlisting}
This proof is unnecessary detailed, and DARe can simplify it to the following:
\begin{lstlisting}
lemma prop_add_comm(x: natr, y: natr)
  ensures add(x, y) == add(y, x)
{
  match x {
   case Zero =>
   case Succ(x') =>
    calc {
      add(x, y);
     ==
      Succ(add(y, x'));
     ==
      add(y, Succ(x'));
    }
  }
}
\end{lstlisting}
The first observation is that the first calculations is not needed so DARe has removed it. For the second calculation, the second step can be removed together with all 
justifications\footnote{This example also illustrates another potential simplification that we do not apply due to our more conservative definitions of annotations. Here, calls
to two auxiliary lemmas are removed. If there are no other calls to them, they can be safely removed. However, as some lemmas may show a useful property and are not used just to 
help a proof we have decided to keep all of them. We may revisit this choice in the future.
}.

\subsubsection{Combining operators in calculations}

A calculation can also have combinations of operators. In some cases, the operator is also omitted, in which the given default operators is used (e.g \scode{calc ==> \{ ...\}}). If no such operator is given, equality is used. 
When removing a line of a calculation, DARe will also remove its preceding operator. This simple solution may result in that we do not find the shortest solution. 
To illustrate, consider the following dummy sequence:
\begin{lstlisting}
  calc {
     A;
   ==>
     B;
   == 
     C;
  }
\end{lstlisting}
DARe may try to remove `\scode{==> B;}'. This will result in 
\begin{lstlisting}
  calc {
     A;
   == 
     C;
  }
\end{lstlisting}
which may not hold while the weaker 
\begin{lstlisting}
  calc {
     A;
   ==> 
     C;
  }
\end{lstlisting}
may hold. This means that with our current design, this simpler solution will not be found.  We plan to revisit this in the future, which will most likely involve using a lattice of the 
supported operators when selecting which operator should be kept.

\bs{}
\section[The DARe tool]{The DARe tool\footnote{The source code 
is available from \cite{DAReWebpage}.}}\label{sec:tool}
\as{}

\begin{figure}
\begin{center}
\includegraphics[width=0.7\textwidth]{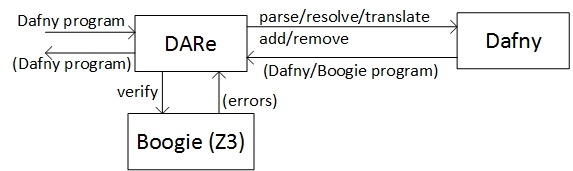}
\end{center}
\caption{Architecture of the DARe tool}\label{fig:arch}
\end{figure}


We have implemented a tool called \emph{DARe} that removes dead annotations as specified in the previous section. The overall architecture of the tool is shown in Figure \ref{fig:arch}. DARe takes as input a Dafny program and returns a new Dafny program where the dead annotations have been removed. It uses Dafny to \emph{parse}, \emph{resolve} and \emph{translate} the program to Boogie. Resolving a program involves some pre-processing and type checking. Dafny stores a program in an AST. DARe will \emph{remove} (and \emph{add} back) annotations from the AST. Dafny returns an updated Dafny program, or a translated Boogie program (if the \emph{translate} command is given). The Boogie program is sent to Boogie for verification, which may return errors. Such errors are analysed and acted upon by DARe as explained next.


We have developed several different algorithms to remove the dead annotations. We give a theoretical comparison of them below, with an empirical evaluation in \S \ref{sec:results}. The first version, called \algname{SimpleDARe},
is as follows:

\begin{algorithmic}
\Procedure{SimpleDARe}{$prog$}
\ForAll{$m \in \Call{Methods}{prog}$}
\State $pos \gets \Call{Start}{m,prog}$
\While{$\Call{hasAnnot}{m,pos}$}
  \State $pre \gets \Call{MethodBody}{m}$
  \State $pos \gets \Call{RemoveNextAnnot}{m,pos,prog}$
  \State $err \gets \Call{Verify}{prog}$
  \If{$\Call{HasError}{err,m}$}
 \State $\Call{ReplaceMethodBody}{prog,m,pre}$
 \EndIf
\EndWhile
\EndFor
\EndProcedure
\end{algorithmic}

\algname{SimpleDARe} steps through the program, method by method. For each method, it finds the initial position. As long as there are more annotations left from a given
position in the method, it will 
remove the next annotation (\algname{RemoveNextAnnot}) from the given position in the given method of the program. This will return the position just after that
annotation.  Note that there will be a side-effect on the program $prog$, where the next annotation has been removed (within method $m$).
It then calls \algname{Verify} to see if Boogie complains. If it reports any errors, then the removal is undone by replacing the body of method $m$ with the body before the annotation was removed
($pre$). This is achieved by \algname{ReplaceMethodBody}. 


Consider a program with the following assertions:
\begin{lstlisting}
assert P;   assert Q;   assert R;
\end{lstlisting}
Assume that the verifier needs either \scode{P}, or \scode{Q} and \scode{R} \big(e.g. \scode{P <==> (Q && R})\big). 
 \algname{SimpleDARe} will remove \scode{assert P} and keep 
\begin{lstlisting}
            assert Q;   assert R;
\end{lstlisting}
The simplest, or shortest, solution would instead have been to keep `\scode{assert P}', as we could then have removed both 
`\scode{assert Q}' and `\scode{assert R}'. 

To overcome such issues, we developed a more complete version, called
\algname{CompleteDARe}. For simplicity, we describe this algorithm recursively:

\begin{algorithmic}
\Procedure{CompleteDARe}{$prog$}
\ForAll{$m \in \Call{Methods}{prog}$}
\State $pos \gets \Call{Start}{m}$
\State $\Call{MethodDARe}{m,pos}$
\EndFor
\EndProcedure  \\

\Procedure{MethodDARe}{$m,pos$}
\If{$\Call{hasAnnot}{m,pos}$} 
  \State $pre \gets \Call{MethodBody}{m}$
  \State $pos \gets \Call{RemoveNextAnnot}{m,pos,prog}$
  \State $err \gets \Call{Verify}{prog}$
  \If{$\Call{HasError}{err,m}$}
  \State $\Call{ReplaceMethodBody}{prog,m,pre}$
  \State $\Call{MethodDARe}{m,pos}$
  \Else 
  \State $\Call{MethodDARe}{m,pos}$
  \State $s_1 \gets \Call{Size}{m,prog}$
  \State $m_1 \gets \Call{MethodBody}{m}$
  \State $\Call{ReplaceMethodBody}{prog,m,pre}$
  \State $\Call{MethodDARe}{m,pos}$
  \State $s_2 \gets \Call{Size}{m,prog}$
  \If{$s_1 < s_2$}
    \State  $\Call{ReplaceMethodBody}{prog,m,m_1}$
   \EndIf
 \EndIf
\EndIf 
\EndProcedure
\end{algorithmic}

The main difference with \algname{SimpleDare} is that when there are no errors, \algname{CompleteDARe} will try both cases of keeping and removing the guidance, 
and keep the version that is shortest. In the above case it would have kept the first assertion and deleted the other two assertions.

The improved completeness achieved by
\algname{CompleteDARe}, results in a drastic increase in (worst case) runtime. We can see that \algname{SimpleDARe} steps through
each guidance one by one, and thus has a $O(N)$ runtime, where $N$ is the number of annotations. For \algname{CompleteDARe}, on the other hand, this is increased to $O(2^N)$ as all combinations are checked.

The increased runtime is undesirable, in particular for our interactive extension, discussed in \S \ref{sec:ide}, where rapid feedback is essential. For
example, in a recent empirical study on the use of program analysis \cite{Christakis16}, the majority of developers ($375$ software engineers at Microsoft) said they would sacrifice intricate issues
for fast feedback, when the analysis is run in real-time (as is the case in \S \ref{sec:ide}). 

We therefore did an experiment by running both \algname{SimpleDare} and \algname{CompleteDARe} on $117$ examples from the Dafny library (see \S \ref{sec:results}) for assertions and invariants. There was only one example where different results were returned, which had the same number of annotations but different ones had been deleted. 
Based on this observation, and the study in \cite{Christakis16}, we decided to sacrifice completeness for speed, and therefore discarded the \algname{CompleteDARe} version.

The main bottleneck for runtime is the time that the verifier uses; this is studied in detail in \cite{Grov2016,tumasthesis16}. In fact, the time DARe (or Dafny) uses to
parse, change, resolve and translate the AST is so insignificant compared with Boogie that it can be ignored.
Whilst Boogie itself has some built in parallelisation \cite{leino2014dafnyIDE}, it is not (currently) possible to apply Boogie in parallel.
Still, one optimisation  is to reduce the number of calls to Boogie, which is achieved by  \algname{CombinedDare}:

\begin{algorithmic}
\Procedure{CombinedDare}{$prog$}
\ForAll{$m \in \Call{Methods}{prog}$}
  \State $pos_m \gets \Call{Start}{m}$
\EndFor
\While{$\exists m \;\cdot\; \Call{hasAnnot}{m,pos_m}$}
  \ForAll{$m \in \Call{Methods}{prog}$}
  \State $pre_m \gets \Call{MethodBody}{m}$
  \State $pos_m \gets \Call{RemoveNextAnnot}{m,pos_m,prog}$
  \EndFor
  \State $err \gets \Call{Verify}{prog}$
  \ForAll{$m \in \Call{Methods}{prog}$}
  \If{$\Call{HasError}{err,m}$}
 \State $\Call{ReplaceMethodBody}{prog,m,pre_m}$
 \EndIf
 \EndFor
\EndWhile
\EndProcedure
\end{algorithmic}


The difference from \algname{SimpleDare} is that each method will in parallel remove a single annotation. It then applies \algname{Verify} (Boogie) to all methods, and keeps the changed methods that do not have any related verification errors. Note that  \algname{RemoveNextAnnot} will do nothing, and just return the same program and position, when the end of the method is reached.

\subsection{Properties of DARe}

\algname{CombinedDare} has no theoretical improvements in worst case runtime from \algname{SimpleDARe}. However, it may in practice as the main bottleneck is
the number of calls to the verifier. Let $\algname{NumAnnots}(m)$ compute the number of annotations for method $m$. We can then summarise the number of calls to the verifier for each algorithm as follows:
\begin{center}
\begin{tabular}{ ccc }
\textbf{Algorithm} & $\quad$& \textbf{Number of calls to verifier} \\
\hline
\algname{SimpleDARe} && $\Sigma_{m \in \algname{Methods}(prog)}{\algname{NumAnnots}(m)}$ \\
\algname{CompleteDARe} && $\Sigma_{m \in \algname{Methods}(prog)}{2^{\algname{NumAnnots}(m)}}$ \\
\algname{CombineDare} && $\algname{Max}\big(\bigcup_{m \in \algname{Methods}(prog)}{\algname{NumAnnots}(m)}\big)$ \\
\end{tabular}
\end{center}
For \algname{SimpleDARe} this is the same as the total number of annotations in the program text. As \algname{CompleteDare} is run method-by-method, the
number of calls is not quadratic on the number of annotations in the program. It is exponential for each method, which is then summed up.  The interesting case is \algname{CombineDare}. As each methods is run in parallel, the number of calls is the number of annotations of the method with the maximum number of annotations.


Below we outline a few simple yet important properties of DARe and the various implementations of the approach. We assume that \algname{RemoveNextAnnot} is correct
with respect to Definition \ref{def:annot}, and that all annotation are tried\footnote{We discuss implementation issues of \algname{RemoveNextAnnot} in \S \ref{sec:tool:design}.}.
We use DARe to cover all algorithms. 
\begin{myprop}
DARe does not change the behaviour or contract of a program. 
\end{myprop}
\begin{proof}
The only changes to a program is done by \algname{RemoveNextAnnot}, which we have assumed will remove annotations. Thus, both the program and its contract is left unchanged.
\end{proof}
DARe will only be applied to verified methods, meaning all stated properties hold. Thus, it cannot make any false properties true (as there are not any). This is also
the case for the dual:
\begin{myprop}
Ignoring runtime environments, DARe will not make any true properties false.
\end{myprop}
\begin{proof}
The only changes DARe makes is to remove an annotation. If a method reports a verification error after removing an annotation then this is added back again to the program.
As a method is assumed to be initially verified, it follows by induction that it will remain verified for each step. 
The only exception is that DARe may result in a longer proof. As the verifier may time-out that means that it may introduce new time-outs on other runtime environments. However, this
case is ignored in the assumption.
\end{proof}
Next we compare the functional properties of the three algorithms.
\begin{myprop}
\algname{CombineDare} and \algname{SimpleDARe} will remove the same annotations. 
\end{myprop}
\begin{proof}
For a given method, both algorithms sequentially steps through the methods and tries to remove the next annotation using \algname{RemoveNextAnnot}. The only difference is that 
\algname{CombineDare} will do this for all methods at the same time to reduce the number of calls to the verifier. The only case where they deviate is when there are no more 
annotations left in a method.  \algname{CombineDare} may then still make calls to  \algname{RemoveNextAnnot}, if there are other methods that have further annotation. However, in this
case  \algname{RemoveNextAnnot} is assumed to behave as an identity, thus the method will be left unchanged.
\end{proof}

One may think that \algname{CompleteDARe} somehow generalises \algname{SimpleDARe} and \algname{CombineDare}. However, this is not the case as it may remove completely different annotations
than the other algorithms. The goal with \algname{CompleteDARe} is to remove as many as possible, thus we have the property:
\begin{myprop}
\algname{CompleteDARe} will remove at least as many annotations as \algname{SimpleDARe} and \algname{CombineDare}. 
\end{myprop}
\begin{proof}
As \algname{CompleteDARe} will try all possibly combinations, including those of  \algname{SimpleDARe} and \algname{CombineDare}, and always pick the result that removes the most annotation, it will
always remove at least as many as the others.
\end{proof}


\subsection{Design choices \& implementation issues}\label{sec:tool:design}

There are variations of how the next annotation is found (by \algname{RemoveNextAnnot}); e.g. we could first remove all assertions, then all invariants, etc, or just remove the next annotation in a sequential manner. We have implemented the latter, although wild-cards in variants and calculations require special care which we will not go into here.

\subsection{Running the DARe tool}

There are two ways to run DARe: 
\begin{enumerate}
\item In the Dafny Visual Studio IDE, which is explained in \S \ref{sec:ide}.
\item As a standalone tool in the command line.
\end{enumerate} 
The advantage of the command line execution is that it can provide various metrics that may be helpful for the user. The command line version has the following syntax:

It can also be used as a command line tool, which can also provide various metrics that may be helpful. Through the command line using the following style for command line parameters: 
\mybox{
\texttt{
> DareTools.exe <Z3 location> <output directory> <list of programs>}
}
The user must provide the location of the Z3 executable (\texttt{<Z3 location>}), the directory where the simplified programs should be stored (\texttt{<output directory>}), and a list of programs
that should be simplified (\texttt{<list of programs>}). The following example illustrates execution:
\mybox{
\begin{tabbing}
\texttt{
> DareTools.exe C:\textbackslash{}Programs\textbackslash{}z3\textbackslash{}z3.exe} \= \texttt{C:\textbackslash{}Dafny} \\ \> \texttt{C:\textbackslash{}Dafny\textbackslash{}Prog.dfy C:\textbackslash{}Dafny\textbackslash{}Progs\textbackslash{}*}
\end{tabbing}
}
When the tool is launched it will tell you what programs were detected so that a user can make sure you entered them correctly.  It will then give the following options:
\begin{enumerate}
 \item \emph{Simplify and print program:}  This will remove all the dead annotations from the program and print a new one.
 \item \emph{Log:} This will remove all the dead annotations from the program and log the details as CVS files. The log will include counting the instances and how many of the different types of annotation is removed.
 \item \emph{Completeness testing:} This will run both \algname{CompleteDare} and \algname{CombineDare} and check if \algname{CombineDare} find the smallest solution.
 \item \emph{Runtime comparison:} This will executes the tree algorithms discussed above, and log the times each takes to run. 
\end{enumerate}
In most cases the first option will be used, and the last third are mainly for evaluation purposes and used for the results discussed in \S \ref{sec:results}.

\bs{}
\section[A semi-automated DARe integrated in the Dafny IDE]{A semi-automated DARe integrated in the Dafny IDE\footnote{A screencast of the DARe IDE is available at \cite{DAReWebpage}.}}\label{sec:ide}
\as{}

There are cases where a developer may want to keep some of the dead annotations as they serve to explain
the code: e.g. a variant to show why a loop terminates, or some key steps in a calculation which explains the correctness
argument. Still, it is likely that many dead annotations should be removed.

To achieve both of these aims, we have developed a \emph{semi-automatic} version of DARe that is integrated into
Dafny's Visual Studio plug-in \cite{leino2014dafnyIDE}. Dafny's IDE adapts the interaction model commonly used by IDE for compilers:
the verifier is applied seamlessly in the background and verification errors are highlighted in the program text, by underlying the 
affected areas by red ``squiggly'' lines, as illustrated on line 31 of Figure \ref{fig:screenshot:highlight} This approach is also recommended in \cite{Christakis16}, where most developers want results
from static analysis shown in the editor.

\begin{figure}[h]
\begin{center}
\includegraphics[width=0.4\textwidth]{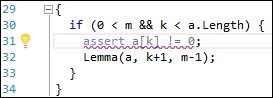} 
\end{center}
\caption{DARe IDE: error highlighting.}\label{fig:screenshot:highlight}
\end{figure}

We have followed the same approach, where DARe is seamlessly applied in the background. When it successfully finishes, all dead annotations will be highlighted 
(underlined and greyed out) as shown on line 31 of Figure \ref{fig:screenshot:highlight}. The user can select which one to remove, and to help automate this she can
press the light-bulb to get the options of removing the given annotation or all annotations of the method or file, as shown in Figure \ref{fig:screenshot:menu}.

\begin{figure}[h]
\begin{center}
\includegraphics[width=0.5\textwidth]{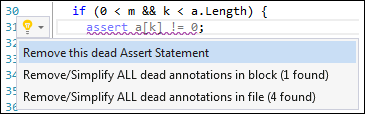}
\end{center}
\caption{DARe IDE: menu to remove dead annotation.}\label{fig:screenshot:menu}
\end{figure}

As DARe will apply the verifier multiple times and is therefore slow compared with Dafny itself, it has been designed to be as unintrusive as possible. It is disabled by
default, and the user can enable it in the Dafny menu\footnote{The main ``pain point'' for static analysis reported in \cite{Christakis16}, is that wrong checks are on by default,
which to some degree justifies our decision to have it disabled by default.}. Figure \ref{fig:screenshot:on} illustrates how to enable it.

\begin{figure}[h]
\begin{center}
\includegraphics[width=0.4\textwidth]{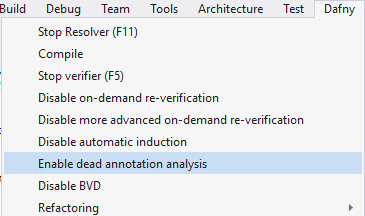} 
\end{center}
\caption{DARe IDE: Enabling Dead Annotation Removal.}\label{fig:screenshot:on}
\end{figure}

Once enabled, DARe will initially be applied to all verified methods after the IDE has remained idle for at least 10 seconds. Once a user starts interacting with the IDE, DARe will terminate (with failure) so that the system does not become 
unresponsive. It will reapply DARe when it has been idle for at least 10 seconds.  The plug-in will keep track of any methods that have changed since last time DARe was applied 
successfully and will reapply DARe to a method only when it has been changed.  Again, this will only happen if the method does not have any verification errors and 
the system remains idle for 10 seconds.

\subsection{Implementation details}

\begin{figure}
\begin{center}
\includegraphics[width=0.45\textwidth]{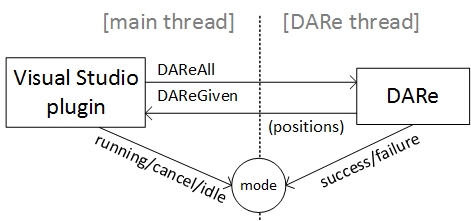}
\end{center}
\caption{Architecture of Visual Studio plug-in integration}\label{fig:idearch}
\end{figure}

The easiest way to achieve the seamless DARe application, is to apply it in a separate thread which is killed when interaction happens. However, this is not possible due
to the issues of applying Boogie in parallel, as elucidated in \S \ref{sec:tool}. Instead, we implemented the architecture shown in Figure \ref{fig:idearch}. We still 
apply DARe in a separate thread, but do not kill it and instead use a shared variable \emph{mode} to control the execution. 

Dafny applies the verifier and waits for the results -- typically with a timeout of 10 seconds. If the program is changed while the verifier runs, it still waits for the
previous verification task to complete before a new one is started \cite{leino2014dafnyIDE}. Our extension has to fit with this model, and works as follows. The plug-in keeps
track of which methods have changed, and when it is decided that DARe should be run, then: either \textsf{DAReAll} is called to apply DARe to all methods (initially);
or \textsf{DAReGiven} is applied, provided with the names of the methods that have changed since last time. It will then set the \textsf{mode} to \textsf{running}. DARe will then execute, and on termination
return the positions in the AST of the annotations that can be removed and set \textsf{mode} to \textsf{success}. If the user starts interacting while DARe is running, then the IDE will set \textsf{mode} to \textsf{cancel}. The \algname{CombinedDare} algorithm of \S \ref{sec:tool} is updated so that for each step, it will check the value of \textsf{mode} before running Boogie: if it is set to \textsf{cancel}, then it will terminate with failure and set the \textsf{mode} to \textsf{failure}. When the plug-in receives the results (either failed or successful), it will set \textsf{mode} back to \textsf{idle}. To avoid conflicts, we have restricted Dafny to only apply the verifier when the mode is \textsf{idle}.

The plug-in consists of two components:
\begin{itemize}
    \item The \emph{tagger} handles the communication with DARe and highlights annotations that can be removed in the program text.
    \item The \emph{action} deals with the actual removal of annotations from the program text when the developer chooses to do so.
\end{itemize}

\bs{}
\section{Evaluation}\label{sec:results}
\as{}


Several examples of the kinds of annotations that can be removed by the DARe tool is shown in \S \ref{sec:dare}. Note that all the removals discussed there has been created by the DARe tool.
In this section we give the results of applying the different algorithms to a larger set of examples. The evaluation focuses on the amount of annotations that can be removed, the tool's recall,
and runtime issues.  The experiments were conducted using Windows 10 with an Intel CORE i5 processor and 16GB RAM.
A complete set of results, including all the programs (before and after DARe was applied) can be found on a dedicated webpage \cite{DAReWebpage}. We first discuss the evaluation data used.


\subsection{Evaluation data \& approach} \label{sec:eval:data}

In order to evaluate the DARe tool we applied it to examples from the Dafny library and test suite \cite{DafnyLib}. The suite combines programs used purely for testing with actual developments. 
The full set is used for regression testing. Developments include solutions to verification challenges and competitions, such as VerifyThis and VSComp. 

We removed all the programs we knew where only used for testing; they are predominantly found in a separate folder in the suite (called \texttt{dafny0}). After this $259$ programs were left. We then
applied DARe to all methods that did not initially report a verification error, which was $7$ methods, leaving $252$ programs. Note that DARe was applied method-by-method: there were several cases within
the $252$ where some methods did not intitially verify. These methods were ignored. 


Note that to compare the algorithm in \S \ref{sec:res:compare} a smaller evaluation set was used as only \algname{CombineDare} supports the full functionality.

\subsection{Annotations removed}

In total there were $2869$ annotation over the programs\footnote{This only includes methods that initially verified.}, meaning on average each program had $11.4$ annotations. DARe managed to remove $88\%$ of the annotations:
a total of $345$ annotations were left, with an average of $1.4$ annotations per program. 

\begin{figure}
\begin{center}
\includegraphics[width=0.6\textwidth]{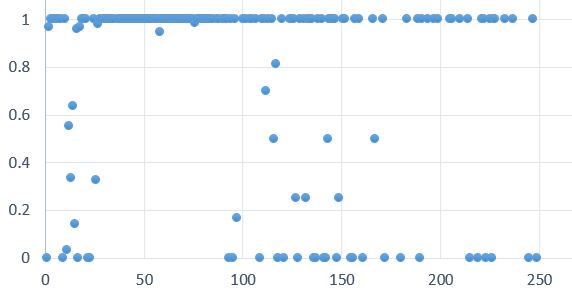}
\end{center}
\caption{Percentage of annotations removed vs LoC.}\label{fig:res:locrem}
\end{figure}

Figure \ref{fig:res:locrem} shows the relationship between annotations removed and the LoC. Note that the relationship between the total number of annotations and annotations removed is almost identical.
As the figure illustrates there is no clear pattern, although it is worth noting that there is a high degree of programs where no annotations were required.

Below we discuss each type of annotations separately.

\subsubsection{Assertions}

In our test set there were $113$ programs with a total of $657$ assertions (average of $5.8$ for each program). DARe was able to remove $91\%$ of these, meaning only $58$ annotations were left (average of $0.5$ for each program). Of the remaining $58$ we tried to remove any conjuncts if present, which could not be done for any of them. 

It is not clear why this number is so high. One of the reasons are that annotations are often used to find a proof or to debug in order to discover where the prover fails. These may then have been left. Another reason is improved automation by Dafny. A third reason is that some of them are probably left to explain the proof and should therefore not be removed.

\subsubsection{Lemma calls}

There were a total of $491$ lemmas calls over $55$ programs in our test-data. Out of these, $85\%$ could be removed leaving only $75$. In many cases all lemma calls could be removed. 

As with assertions, the number of lemma calls that can be removed is surprisingly high. Some of the reasons are the same as for assertions, as one can see lemma calls as ways of asserting richer properties that require proofs.
Another direct reasons may be the introduction of an ``induction tactic'' \cite{Leino12a}, which means that many recursive function calls are no longer required. Any development pre-dating this feature would have had such calls
which can now be removed.

\subsubsection{Invariants}

$57$ programs had a total of $393$ invariants. DARe was able to remove $102$ of these invariants, meaning $26\%$ could be removed. As with assertions, invariants may be made of a conjunction. For the remaining $291$ invariants we could split some of them, resulting in $125$ conjuncts\footnote{We only consider the invariants that could be split up.} Of these $6$ conjuncts could be removed.

\subsubsection{Variants}

$51$ programs had a total of $168$ variants, where $40$ ($24\%$) could be removed. 

As we have already mention, no wild-card variants could not be removed: $6$ programs used a total of $16$ of these clauses and none of them could be removed.  

\subsubsection{Calculations}

$29$ programs had a total of $114$ calculations. $89\%$ ($101$) could be removed fully. $8$ of the remaining calculations where in one program (\scode{Calculations.dfy}). When splitting them up as described above we ended up with $47$ parts where $38\%$ of them could be removed. The remaining calculations where unchanged. 

Again, this is a surprisingly high number that can be removed. One reason may be that calculations are predominantly used `explain' proofs and not to automate them.

\subsection{Runtime} 

The most important feature was if removing annotations had any impact on runtime. There are reasons
that it may increase of decrease:
\begin{itemize}
\item By removing unnecessary clutter, the prover can find a proof faster, thus reducing the verification time;
\item By removing annotations, we could remove a shorter proofs, thus increasing the search required and thus increasing verification time.
\end{itemize}

To evaluate runtime we ran each program $3$ times and used the average from this. On average for all $252$ programs, each program took $1027$ ms to verify before annotations were removed, and $566$ ms
afterwards, i.e. the runtime was almost halved. The average reduction for all programs was $461$ ms, while the median was only $4.5$ ms. From that one may conclude that many programs were either very small or there were a few programs with a drastic improvement in runtime while very few had a massive improvement.










\subsection{Comparing algorithms}\label{sec:res:compare}

In \S \ref{sec:tool} three different algorithms: \algname{CombineDare} is the one the tool supports by default (and used by the GUI), while \algname{CompleteDare} tries all possible combinations and will find the
solution with fewest annotations.  For technical reason we were only able to compare the algorithms for a version of DARe that does not have all the features presented above. This version was only applicable to $117$ programs, and is the same data as used in \cite{2016NWPT}.

Our first experiment was to test the recall of \algname{CombineDare} by comparing it with an implementation of \algname{CompleteDare}. As this was only for comparison we only 
support assertions and invariant removal.

It was only case where the results from the two algorithms different, and in that case the size was actually the same. Although this may suggest a recall of $100\%$, the size of the data and the fact we only support
assertions and invariants makes such conclusion premature. However, it was sufficient for us to conclude that in practice the recall of \algname{CombineDare} is sufficient to justify using it.

\begin{figure}
\begin{center}
\includegraphics[width=0.7\textwidth]{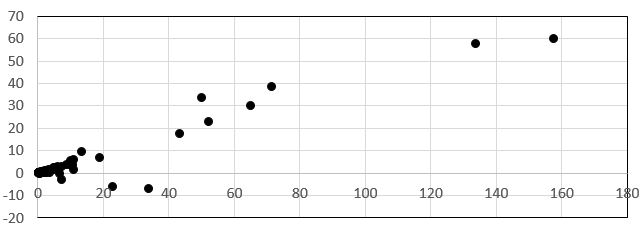}
\end{center}
\caption{Runtime difference between \algname{SimpleDARe} and \algname{CompleteDare}.}\label{fig:runtime:diff}
\end{figure}

To compare runtime we only compared \algname{SimpleDARe} and \algname{CombinedDare} and \algname{CompleteDARe} only 
supports assertions and invariants so a comparison will neither be fair nor informative. 
The difference in runtime between \algname{SimpleDARe} and \algname{CombinedDare} was not as large as expected. The average difference is about $5$ seconds, with a median of only $0.26$ seconds.  Figure \ref{fig:runtime:diff} shows the difference in runtime (Y-axis) with respect to the overall runtime of \algname{SimpleDARe}. The figure shows that while the difference increases with the runtime, there is no indication that the factor is related to runtime. The smaller than expected difference may be due to:
an uneven distribution of annotations between the methods\footnote{Recall that \algname{CombinedDare} performs better the more distributed annotations are.} and optimisations within Boogie
where verification results are cached at the method level so that only changed methods are re-verified \cite{leino2015fine}.

\bs{}
\section[Related work]{Related work\footnote{An abstract of this work appears in \cite{2016NWPT}, and a tool demo is given in \cite{2016FIDE}. \cite{2016NWPT} only gives a high-level view of the tool and results, while \cite{2016FIDE} only illustrates tool usage. This is a special issue dedicated to the work presented in \cite{2016FIDE}.}}\label{sec:related}
\as{}

This work is an example of \emph{search based software engineering} \cite{harman2012search}, where search techniques are used to support and automate the work by software engineers. 
The problem of discovering which annotations are unnecessary, which is the key challenge of our work, is comparable to the \emph{credit-assignment problem} 
in reinforcement learning \cite{minsky63}. This is the problem of  determining which components contribute to the overall success of a system or game.

We are not familiar with any similar work for program verifiers, where the focus tends to be on the \emph{discovery} of auxiliary annotations (e.g \cite{Gupta09,hoder2011invariant,srivastava2009program,ernst2007daikon}).
Our work is inspired by \emph{dead code removal} (or dead code elimination), 
which is the problem of discovering and removing code that is never reached. This is a common compiler optimization, and is typically achieved by data-flow analysis.  

Within automated theorem proving there is a large line of work concerned with selecting only necessary lemmas and axioms provided to the prover in order to reduce the search space (e.g. \cite{sutcliffe2007srass,kuhlwein2012overview}).
An alternative approach to DARe would be to see if any of these heuristics can be explored to automatise our search, which currently remove annotations one by one. Another approach 
is to derive the required annotations from the proof terms generated by Z3 \cite{deMouraZ3roofs08}, and remove any annotation not used. This will not require any additional Boogie calls, and should thus be considerably faster and scale better. This approach is however complicated by the way in which translation from Dafny into Z3 is handled via Boogie. As our work was mainly motivated by discovering the amount of dead annotation, we decided for the simpler search-based approach. In the future, it would be interesting to see how feasible it is to analyse the result from Z3 (which may also be beneficial for general caching of verification attempts \cite{leino2015fine}). Note that this is similar to how 
the \emph{Sledgehammer} tool for Isabelle uses tools like Z3 in order to reconstruct the proofs within Isabelle \cite{paulson2010three}.

Dual to our work of trying to a find a minimal solution to a success criteria (i.e. provability), is finding the minimal causes of failures. This is particularly relevant
for debugging regressions. Two approaches that are reminiscent of our work are: \emph{regression containment} \cite{Ness97}, which automates the process of isolating regressions;
and the more general \emph{delta debugging} \cite{Zeller1999yesterday}. Here, the smallest configuration that introduces a regression is (automatically) sought, which is similar
to the manner in which we are seeking the smallest set of annotations. ``The method'' of the ACL2 community \cite{kaufmann2013computer}
involves deriving a ``toy problem'' from a problem such that the solution of the toy problem can be ported to solve the main problem. One would like this toy problem 
to be as small as possible. Note that this is a (manual) development method, and has not been automated.

The topic of \emph{automated repair} has recently received a great deal of attention, with a large number of tools addressing this problem\footnote{See \url{automated-program-repair.org}.}.  
The  SemFix repair tool \cite{Nguyen:2013:SemFix} steps through each statement, ranked by their suspiciousness of being the bug, and attempts to repair them one-by-one 
in the same way as DARe steps through annotations one-by-one (albeit without any ranking). 
The AutoFix tool \cite{Pei2014} repairs programs in presence of contracts, and uses program synthesis techniques to add statements. This is again analogous to how DARe works, 
with the difference that DARe removes annotations as opposed to generating statements.
 Within model-driven engineering there has been work on discovering redundant constraints, i.e. constraints that can be derived from other constraints
\cite{Wang15}. This work uses the Alloy system \cite{jackson2012software}  to derive these, and is similar to DARe in that redundant constructs are removed using a verification system.

The main advantage of our approach relates to readability, as there is little effect on runtime once redundant annotations are removed. There are many tools to support developers in achieve ``better code'' in that respect. The most well-known is probably \emph{ReSharper} from JetBrains\footnote{See \url{www.jetbrains.com/resharper/}.}. As with DARe, it has a Visual Studio plug-in, and can highlight issues related to quality (e.g. redundant casts) as well as use of coding standards. 

There are also several \emph{code refactorings} that relate to improved readability \cite{Fowler99}, and an additional, or even alternative, way to improve readability and structure is to apply more general refactorings to the annotations in the program text. Refactorings have previously been studied for \emph{proof scripts} \cite{whitesidethesis,Whiteside11}, and we have started to address this for program verifiers. However, it is not clear how this can be fully automated.
Recent work on tactics for Dafny aims to provide more high-level guidance, and attempts to generate a minimum number of annotations from given verification patterns \cite{Grov2016,Grov2016FM}.


\bs{}
\section{Conclusion \& future work}\label{sec:concl}
\as{}

We have presented the DARe tool, which automatically removes annotations that are not required by the verifier. We have extended the Dafny IDE with a semi-automatic version of the tool,
and shown that in our evaluation set, $88\%$ of annotations can be removed. It was particularly successful for assertions, lemma calls and calculations.

Next, we would like to apply DARe to larger programs.
In particular, the large Ironfleet development, which in total has $40K$ lines of ``proofs''\footnote{This number is likely to contain elements which we have not classified as annotations.} \cite{hawblitzel2015ironfleet}, would be a very interesting case study. This will stress-test DARe, but is likely to require some special-purpose configurations to DARe (e.g. change of time-out values), due to the sheer size of the development. It is also worth noting that most of the programs in the Dafny library have been developed by the main Dafny developer (Rustan Leino), who is considered to be the main expert on Dafny. He is therefore likely to write more ``minimal'' annotations than ``normal'' users. It will therefore be interesting to see if the number of dead annotations increase when developed by others.


There are several ways to extend the work further. We could address further simplification techniques, possibly normalising formulas to CNF form and then split the conjuncts. We could also extend the tool to other constructs. For example, programming constructs for the ghost state can also be removed. Another example is to remove/simplify contracts for (auxiliary) recursive methods which are analogous to loop invariants; in this case a user may need to indicate whether or not a lemma is auxiliary. Similarly, certain auxiliary lemmas that are no longer called can be removed. Another possibility, is to extend the work to more meta-level analysis such as subsumption and equivalence relations between annotations, as well as exploring the Z3 proof terms as discussed in \S \ref{sec:related}.

We would like to empirically investigate if the classification of annotations makes sense for Dafny developers.  
This could be achieved by interviews, user experiments, and (to some degree) by instrumenting DARe with a logging mechanism to capture which dead annotations a developer selects to remove. We may also make the classification user configurable. Such studies will help inform  the development of future releases of DARe.  Finally, we could implement similar tools for other program verifiers and re-do the same experiments for them.



\bibliographystyle{plain}
\bibliography{dafnybibl}

\end{document}